\documentclass[submission]{eptcs}
\providecommand{\event}{WPTE 2016} 
\usepackage{breakurl}             
\usepackage{latexsym}
\usepackage{latexsym}
\usepackage{amsmath}
\usepackage[psamsfonts]{amssymb}
\usepackage[all]{xy}
\DeclareMathAlphabet{\mathbbold}{U}{bbold}{m}{n}
\usepackage[varg]{txfonts}

\newcommand*{\interpreted}[1]{{[\![ #1 ]\!]}}
\newcommand*{\eqclass}[1]{{[#1]_\equiv}}
\newcommand*{\ov}[1]{{\overline{#1}}}
\def\Epsilon{\mathcal{E}}

\def\rank{\mathrm{rank}}
\def\sig{\mathrm{sig}}
\def\iff{\mathrel{\mathrm{iff}}}

\newtheorem{theorem}{Theorem}
\newtheorem{definition}[theorem]{Definition}
\newtheorem{lemma}[theorem]{Lemma}
\newtheorem{proposition}[theorem]{Proposition}
\newtheorem{corollary}[theorem]{Corollary}
\newtheorem{example}[theorem]{Example}
\newenvironment{proof}[1][Proof]{\begin{trivlist}
\item[\hskip \labelsep {\bfseries #1}]}{\end{trivlist}}

\providecommand{\squareforqed}{\hbox{\rlap{$\sqcap$}$\sqcup$}}
\providecommand{\QED}{{\unskip\nobreak\hfil%
	\penalty50\hskip1em\null\nobreak\hfil\squareforqed%
	\parfillskip=0pt\finalhyphendemerits=0\endgraf}}

\def\<#1>{\langle #1 \rangle}  

\title{An Extension of Proof Graphs for Disjunctive Parameterised Boolean Equation Systems}
\author{Yutaro Nagae
\institute{\quad Graduate School of Information Science\\
\quad Nagoya University}
\email{nagae\_y@trs.cm.is.nagoya-u.ac.jp}
\and
Masahiko Sakai
\institute{\quad Graduate School of Information Science\\
\quad Nagoya University}
\email{\quad sakai@is.nagoya-u.ac.jp}
\and
Hiroyuki Seki
\institute{\quad Graduate School of Information Science\\
\quad Nagoya University}
\email{\quad seki@is.nagoya-u.ac.jp}
}
\def\titlerunning{An Extension of Proof Graphs for Disjunctive PBESs}
\def\authorrunning{Yutaro Nagae, Masahiko Sakai \& Hiroyuki Seki}
\begin{document}
\maketitle

\begin{abstract}
 A parameterised Boolean equation system (PBES) is a set of equations
 that defines sets as the least and/or greatest fixed-points
 that satisfy the equations.
 This system is regarded as a declarative program defining functions that take a datum and returns a Boolean value.
 The membership problem of PBESs is a problem to decide whether a given element is in the defined set or not,
 which corresponds to an execution of the program.
 This paper introduces reduced proof graphs, and studies a technique to solve
 the membership problem of PBESs, which is undecidable in general, by transforming it into a reduced proof graph.

 A vertex $X(v)$ in a proof graph represents that the data $v$ is in
 the set $X$, if the graph satisfies conditions induced from a given
 PBES.  Proof graphs are, however, infinite in general.  Thus we
 introduce vertices each of which stands for a set of vertices of the
 original ones, which possibly results in a finite graph.
 For a subclass of disjunctive PBESs, we clarify
 some conditions which reduced proof graphs should satisfy.
 We also show some examples
 having no finite proof graph except for reduced one.
 We further propose a reduced dependency space, which contains reduced proof graphs as sub-graphs if a proof graph exists.
 We provide a procedure to construct finite reduced dependency spaces,
 and show the soundness and completeness of the procedure.
\end{abstract}

\section{Introduction}
A \emph{Parameterised Boolean Equation System} (PBES)~\cite{Groote2004,GROOTE2005332} is a set of equations
denoting some sets as the least and/or greatest fixed-points.
PBESs can be used as a powerful tool for solving a variety of problems
such as process equivalences~\cite{Chen:2007:ECI:2392200.2392211}, 
model checking~\cite{Groote:2005:MPD:1099102.1099103}, and so on.
 
We explain PBESs by an example PBES $\Epsilon_1$, which consists of the following two equations:
\[
 \begin{array}[t]{rcl}
  \nu X(n:N) & = & X(n+1) \vee Y(n)\\
  \mu Y(n:N) & = & Y(n+1)
 \end{array}
\]
$X(n:N)$ denotes that $n$ is a natural number and a formal parameter of $X$.
Each of the predicate variables $X$ and $Y$ represents a set of natural numbers regarding that $X(n)$ is true if and only if $n$ is in $X$.
These sets are determined as fixed-points that satisfy the equations,
where $\mu$ (resp.\ $\nu$) represents the least (resp.\ greatest) fixed-point.
In the PBES $\Epsilon_1$, $Y$ is an empty set since $Y$ is the least fixed-point
satisfying that 
$Y(n) \iff Y(n+1)$ for any $n\geq 0$.
Similarly, $X$ is equal to $\mathbb{N}$ since $X$ is the greatest fixed-point
satisfying that 
$X(n) \iff X(n+1) \vee Y(n)$ for any $n\geq 0$.

The membership problem for PBESs is undecidable in general, and some techniques
have been proposed to solve the problem for some subclasses of PBESs:
one by instantiating a PBES to a
\emph{Boolean Equation System} (BES)~\cite{PLOEGER2011637} and
one by constructing a proof graph~\cite{Cranen2013}.
In the latter method, the membership problem is reduced to an existence of a proof graph.  A proof graph that justifies $X(0)$ for $\Epsilon_1$ is shown as follows:
\[
 \xymatrix{
   X(0) \ar[r] & X(1) \ar[r] & X(2) \ar[r] & \cdots,
 }
\]
where each vertex $X(n)$ represents that the predicate $X(n)$ holds.
If there exists a finite proof graph for a given instance of the problem,
it is not difficult to 
find it
mechanically.
However, finite proof graphs do not always exist.

In this paper, we extend proof graphs and propose \emph{reduced} proof graphs, where vertices stands for a set of vertices in the original ones.
We clarify some conditions which reduced proof graphs for data-quantifier free and disjunctive PBESs should satisfy,
where data-quantifier free and disjunctive PBESs are a subclass of disjunctive PBESs~\cite{Koolen2015}.
We also provide a \emph{reduced} dependency space and show that it contains reduced proof graphs as sub-graphs if a proof graph exists.
We give a procedure to construct a finite reduced dependency space,
and show soundness and completeness of the procedure.
We also show examples having no finite proof graph but finite reduced ones.

\section{PBESs and Proof Graphs}
We follow~\cite{Cranen2013} and~\cite{GROOTE2005332} for basic notions related to PBESs and proof graphs.

We assume non-empty \emph{data sorts}. 
For every data sort $D$, we 
assume a set $\mathcal{V}_D$ of \emph{data variables} and a \emph{semantic domain} $\mathbb{D}$ corresponding to it.
In this paper, we assume the existence of a sort $B$ corresponding to the Boolean domain $\mathbb{B} = \{\mathbbold{t},\mathbbold{f}\}$
and a sort $N$ corresponding to the natural numbers $\mathbb{N}$.
A \emph{data environment} $\delta$ is a function 
that maps each data variable to a value of the associated type.
A \emph{data environment update} $\delta[v/d]$ for a data variable $d$ of a sort $D$ and $v\in\mathbb{D}$ is a mapping defined by $\delta[v/d](d') = v$ if $d=d'$ and $\delta[v/d](d')=\delta(d')$ otherwise.
We assume appropriate \emph{data functions} on $\mathbb{D}$, and
use $\interpreted{e}\delta$ to represent a value in $\mathbb{D}$
obtained by the evaluation of a \emph{data expression} $e$ of a sort $D$
under a data environment $\delta$.
A data expression interpreted to a value in $\mathbb{B}$ is called a \emph{Boolean expression}.
In this paper, we use usual operators and constants like $\mathrm{true}$, $\mathrm{false}$, $\le$, $0$, $1$, $+$, $-$, and so on, as data functions in examples without stating.

A \emph{Parameterised Boolean Equation System} (PBES) is a set of equations defined as follows.  The syntax of PBESs is 
given through the following grammar:
  \[
  \begin{array}[t]{rcl}
   \Epsilon & ::= & \emptyset \ |\
	\left(\nu X\left(d: D\right) = \varphi\right) \Epsilon\ |\
	\left(\mu X\left(d:D\right) = \varphi\right) \Epsilon \\
   \varphi & ::= & b \ |\ \varphi \wedge \varphi \ |\ \varphi \vee
	\varphi \ |\ \forall d{:}D\ \varphi \ |\ \exists d{:}D\ \varphi \ |\ X(e)
  \end{array}
  \]
Here, $\emptyset$ is used for the \emph{empty PBES}, and quantifiers $\mu,\nu$ are used to indicate the least and greatest fixed-points, respectively.
$\varphi$ is a \emph{predicate formula},
$X$ is a predicate variable sorted with $D \rightarrow B$,
$b$ is a Boolean expression,  $d$ is a data variable of a sort
$D$, and $e$ is a data expression. 

  A PBES is regarded as a sequence of equations,
  \[ \Epsilon\ =\
  \left(\sigma_1 X_1\left(d: D\right) = \varphi_1\right) \mathrel{\cdots}
  \left(\sigma_n X_n\left(d: D\right) = \varphi_n\right)
  \]
  where $\sigma_i \in \{\mu,\nu\}$ ($1 \leq i \leq n$).
We say $\Epsilon$ is \emph{closed} if it contains no free predicate
variables as well as no free data variables.
Note that the negation is allowed only in expressions
$b$ or $e$ as a data function.

 \begin{example} \label{eg:PBES}
  A PBES $\Epsilon_2$ is given as follows:
  \[
  \begin{array}[t]{rcl}
   \nu X_1(d: N)& = & (\mathrm{true} \wedge X_1(d+1))\ \vee\ \left(d \geq 1 \wedge X_2\left(d\right)\right)\\
   \mu X_2(d: N)& = & (\mathrm{true} \wedge X_2(d+1))\ \vee\ \left(d \mathrel{=} 0 \wedge X_1\left(d\right)\right)
  \end{array}
  \]
 \end{example}
Obviously the occurrences of '$\mathrm{true}$' are redundant in the expressions.
They are necessary for the subclass introduced in Section~\ref{sec:ePG}.

Since the definition of the semantics is complex, we will give an intuition by an example before introducing the formal definition.
The meaning of a PBES is determined in the bottom-up order.
Considering a PBES $\Epsilon_2$ in Example~\ref{eg:PBES},
we first look at the second equation, which defines a set $X_2$.
The set $X_2$ is fixed depending on the free variable $X_1$, i.e., the equation should be read as that $X_2$ is the least set satisfying the following condition for any $v\in \mathbb{N}$:
\[
 v \in X_2\ \iff\ v+1 \in X_2\ \vee\ (v=0\wedge v\in X_1).
\]
Thus the set $X_2$ is fixed as
\[  X_2 = 
   \begin{cases}
        \{0\}      & \text{if $0\in X_1$}\\
         \emptyset & \text{otherwise}
   \end{cases},
\]
i.e., $X_2(v) \iff X_1(v)\wedge v=0$ for any $v\in \mathbbold{N}$.
Next, we replace the occurrence of $X_2$ in the first equation, which results in the following equation:
\[
   \nu X_1(d: N)\ =\  X_1(d+1)\ \vee\ \left(d \geq 1 \wedge (X_1\left(d\right) \wedge d=0)\right).
\]
Since this is simplified as $\nu X_1(d: N)\ =\  X_1(d+1)$, the set
$X_1$ is fixed as the greatest set satisfying that $v \in X_1\ \iff\ v+1 \in X_1$ for any $v\in \mathbb{N}$.
All in all, we obtain $X_1=\mathbb{N}$ and $X_2=\{0\}$.
This is formally defined~\cite{GROOTE2005332} as shown below.

We assume a \emph{predicate environment}
$\theta : \mathcal{P} \rightarrow (\mathbb{D} \rightarrow \mathbb{B})$
for a set $\mathcal{P}$ of predicate variables, i.e.,
$\theta$ assigns a function to each predicate variable.
We define a \emph{predicate environment update} $\theta [f/X]$ in a similar way to a data
environment update.
The semantics of a predicate formula $\varphi$ is defined as follows:
\[
 \begin{array}{rcl}
  \interpreted{b}\theta\delta &=& \interpreted{b}\delta\\
  \interpreted{X(e)}\theta\delta & = &
   \theta(X)(\interpreted{e}\delta) \\
  \interpreted{\varphi_1 \oplus \varphi_2}\theta\delta & = &
   \interpreted{\varphi_1}\theta\delta \oplus
   \interpreted{\varphi_2}\theta\delta \\
  \interpreted{\diamond d{:}D\ \varphi}\theta\delta & = &
   \diamond v \in \mathbb{D}\ \interpreted{\varphi}\theta\delta[v/d]
 \end{array}
\]
where $\oplus \in \{\vee,\wedge\}$ and $\diamond = \{\forall, \exists\}$.

 \begin{definition}
  For a PBES $\Epsilon$, a predicate environment $\theta$, and a data environment $\delta$, the tuple $\left< \Epsilon, \theta, \delta\right>$ is an \emph{interpreted PBES}.
  The \emph{solution} of an interpreted PBES 
  is a predicate environment $\interpreted{\Epsilon}\theta\delta$ determined by the interpretation defined as follows:
  \[
   \begin{array}{rcl}
	\interpreted{\emptyset}\theta\delta & = & \theta\\
	\interpreted{(\sigma X(d:D) = \varphi)\Epsilon}\theta\delta & = &
	 \interpreted{\Epsilon}\theta[\sigma T/X]\delta\\
   \end{array}
  \]
  where $\sigma \in \{\mu, \nu\}$ and $T : (\mathbb{D} \rightarrow \mathbb{B}) \rightarrow (\mathbb{D}
  \rightarrow \mathbb{B})$
  is the \emph{predicate transformer} defined by
  \[
  T = \lambda f\in \mathbb{B}^\mathbb{D}.\lambda v \in \mathbb{D}.
  \interpreted{\varphi}\left(\interpreted{\Epsilon}\theta[f/X]
  \delta\right)\delta[v/d].
  \]
 \end{definition}
 Note that the solution does not depend on the environments $\theta$ or $\delta$ if the system is closed.

 \begin{example} \label{eg:PBESsol}
  For $\Epsilon_2$ given in Example~\ref{eg:PBES},
    the solution 
$\interpreted{\Epsilon}\theta\delta$ is characterized such that
    $(\interpreted{\Epsilon}\theta\delta)(X_1)$
          (resp.\ $(\interpreted{\Epsilon}\theta\delta)(X_2)$)
    is a function that
    values $\mathbbold{t}$ if and only if an arbitrary natural number (resp.\ $0$) is given.
 \end{example}

The \emph{membership problem} for PBESs is a problem that answers whether $X(d)$ holds or not for a given interpreted PBES and $X$ and $d$.
In the sequel,
we explain proof graphs introduced in~\cite{Cranen2013} in order to characterize the membership problem.

For a PBES 
  \( \Epsilon\ =\
  \left(\sigma_1 X_1\left(d: D\right) = \varphi_1\right) \mathrel{\cdots}
  \left(\sigma_n X_n\left(d: D\right) = \varphi_n\right),
  \)
the \emph{rank} of $X_i$ ($1\leq i\leq n$) is the number of alternations of $\mu$ and $\nu$
in the sequence $\nu\sigma_1\cdots\sigma_n$.
Note that the rank of $X_i$ bound with $\nu$ is even
and the rank of $X_i$ bound with $\mu$ is odd.
For Example~\ref{eg:PBES}, $\rank_{\Epsilon_2}(X_1)=0$ and $\rank_{\Epsilon_2}(X_2)=1$.
\emph{Bound variables} are predicate variables $X_i$
that occur in the left-hand sides of equations in $\Epsilon$.
The set of bound variables are denoted by $\mathrm{bnd}(\Epsilon)$.
The \emph{signature} $\sig(\Epsilon)$ in $\Epsilon$ is defined by
$\sig(\Epsilon) = \left\{(X_i,v) \mid X_i \in \mathrm{bnd}(\Epsilon),\ v \in \mathbb{D} \right\}$.
We use $X_i(v)$ to represent $(X_i,v)\in \sig(\Epsilon)$.
We use the notation $u^{\bullet}$ for the post set $\{u' \in V \mid u \rightarrow u'\}$
 of a vertex $u$ in a directed graph $\left<V,\rightarrow\right>$.
 \begin{definition}\label{def:PG}
  Let $\left<\Epsilon, \theta, \delta\right>$ be an interpreted PBES,
  $V \mathrel{\subseteq} \sig(\Epsilon)$, ${\rightarrow} \mathrel{\subseteq} {V \times V}$, and $r \in \mathbb{B}$.
  If both of the following conditions hold for any $X_i(v) \in V$,
  the tuple $\left< V, \rightarrow, r \right>$ is called a \emph{proof graph} for the PBES.
  \begin{enumerate}
   \item\label{def:PG:1} 
	$\interpreted{\varphi _i} (\theta [\lnot r / \sig(\Epsilon)][r / X_i(v)^{\bullet}]) (\delta [v / d]) = r$
   \item\label{def:PG:2}
	For any infinite sequence $Z_0(x_0) \rightarrow Z_1(x_1) \rightarrow \cdots$ that begins from $X_i(v)$,
	the minimum rank of $Z^{\infty}$ is even,
	where $Z^{\infty}$ is the set of $Z_i$ that occurs infinitely often in the sequence. 
  \end{enumerate}
 \end{definition}
 The condition (\ref{def:PG:1}) says that $\varphi_i = r$ must hold if we assume that the successors of $X_i(v)$ are $r$ and
 the other signatures are $\lnot r$.

 We say that a proof graph $\left<V, \rightarrow, r\right>$ \emph{proves} $X_i(v) = r$
 if and only if $X_i(v) \in V$.
In the sequel, we assume $r = \mathbbold{t}$. 
   \begin{example} \label{eg:PG}
	 Consider the following graph and $\Epsilon_2$ in Example~\ref{eg:PBES}:
	\[
	\xymatrix{
	X_1(0) \ar[r] & X_1(1) \ar[r] & X_1(2) 
	\ar[r] & \cdots\\
	X_2(0) \ar[u]
	}
	\]
	 This graph is a proof graph proving $X_2(0)=\mathbbold{t}$, which is justified as follows.
	 We have that if $X_1(0) = \mathbbold{t}$ then $X_2(0) = \mathbbold{t}$,
	 and if $X_1(n+1) = \mathbbold{t}$ then $X_1(n) = \mathbbold{t}$ for any $n \geq 0$.
	 Therefore, this graph satisfies the condition (\ref{def:PG:1}) in Definition~\ref{def:PG}.
	 Moreover, $X_1$ occurs infinitely often in an infinite path in the graph and rank of $X_1$ is even.
	 Thus, the condition (\ref{def:PG:2}) is satisfied.
   \end{example}

   The next theorem states the relation between proof graphs and the membership
   problem on a PBES.
 \begin{theorem}[\cite{Cranen2013}]\label{thm:PBES2PG}
  For an interpreted PBES $\left<\Epsilon,\theta,\delta\right>$ and
  a $X_i(v) \in \sig (\Epsilon)$, 
  the existence of a proof graph $\left<V,\rightarrow,r\right>$ 
  such that $X_i(v) \in V$ coincides with
  $\interpreted{X_i(v)}\theta\delta = r$.
 \end{theorem}

\section{Reduced Proof Graphs}\label{sec:ePG}
This section extends proof graphs, called \emph{reduced proof graphs},
in which each vertex is a set of vertices with the same predicate symbol
in the original proof graphs. 
We write a vertex as $X_i(C)$, which stands for $\{X_i(v) \mid v \in C \subseteq \mathbb{D}\}$.
We begin with an example.
\begin{example}\label{eg:ePG}
 A reduced proof graph for $\Epsilon_2$ in Example~\ref{eg:PBES} is shown as follows:
 \[
 \xymatrix{
   X_1(\{0\}) \ar[r] & X_1(\{d \mid d \geq 1\}) \ar@(rd,ld) \\
   X_2(\{0\}) \ar[u] &
 }
 \]
 The vertices $X_1(\{0\})$ and $X_2(\{0\})$ naturally correspond to
 $X_1(0)$ and $X_2(0)$ in Example~\ref{eg:PG}, respectively.
 On the other hand, the vertex $X_1(\{d \mid 1 \leq d\})$ represents
 the infinite set of vertices $\{X_1(1), X_1(2), \ldots\}$.
\end{example}

For consistency, an edge from a vertex $X_i(C)$ to a vertex $X_j(C')$
is allowed, if for any $v\in C$ there exists $v'\in C'$ such that the
edge from $X_i(v)$ to $X_j(v')$ meets the condition (\ref{def:PG:1}) in
Definition~\ref{def:PG}.  This is the main difference with the original definition.

In the rest of this paper, we focus on a restricted class of PBESs,
where the graph construction in Section~\ref{sec:semi-alg} makes sense under such a restriction.
\begin{definition}\label{def:dq-free-dis}
A closed PBES is \emph{data-quantifier free and disjunctive} if it is in the following forms:
 \[
 \begin{array}[t]{rcl}
  \sigma_1 X_1(d: D)& = & {\displaystyle \bigvee_{1 \leq k \leq m_1}}
   \Bigl( \varphi_{1k}(d) \wedge X_{a_{1k}}\bigl(f_{1k}(d)\bigr) \Bigr) \\
  & \vdots &\\
  \sigma_n X_n(d: D)& = & {\displaystyle \bigvee_{1 \leq k \leq m_n}}
   \Bigl( \varphi_{nk}(d) \wedge X_{a_{nk}}\bigl(f_{nk}(d)\bigr) \Bigr) \\
 \end{array}
 \]
where
$f_{ik}(d)$ is a data expression possibly containing variable $d$,
and $\varphi_{ik}(d)$ is a predicate formula, defined by the grammar
$\varphi ::= b \ |\ \forall d'{:}D\ \varphi \ |\ \exists d'{:}D\ \varphi$,
containing no free variables except for $d$.
\end{definition}
$\Epsilon_2$ in Example~\ref{eg:PBES} is data-quantifier free and disjunctive.
This class is a subclass of disjunctive PBESs introduced in~\cite{Koolen2015}.
In disjunctive PBESs, the right hand sides of the equations are in the following form:
\[
  {\displaystyle \bigvee_{1 \leq k \leq m_n}} \exists e : E_k
  \Bigl( \varphi_{ik}(d,e) \wedge X_{a_{ik}}\bigl(f_{ik}(d,e)\bigr) \Bigr)
\]
Here, the value of $f_{ik}(d,e)$ satisfying $\varphi_{ik}(d,e)$ varies according to the value of $e$ for a parameter $d$.
From the restriction ``data-quantifier free'', we get the unique $f_{ik}(d)$
for a parameter $d$.
We use this fact to argue reduced proof graphs.
  For closed PBESs, we abbreviate $\interpreted{\Epsilon}\theta\delta$ as $\interpreted{\Epsilon}$.

 PBESs in the subclass inherit the following important property that holds for the disjunctive \linebreak PBESs~\cite{Koolen2015}.
 \begin{proposition}\label{prop:PG-loop}
  For a data-quantifier free and disjunctive PBES $\Epsilon$ and an $X_i(v) \in \sig(\Epsilon)$, 
  the property $\interpreted{\Epsilon}(X_i)(v) = \mathbbold{t}$ coincides with
  the existence of a proof graph 
  such that $X_i(v) \in V$ and $|w^{\bullet}| \mathrel{=} 1$
  for every vertex $w$ in the graph.\footnote{
It is stated that $|w^{\bullet}| \leq 1$ in~\cite{Koolen2015}.  
  The equality is, however, easily derived from the disjunctivity.
  If $w$ is in a proof graph, then $|w^{\bullet|} \geq 1$ must hold
  from the form of disjunctive PBES
  to satisfy the condition~(\ref{def:PG:1}) of the proof graph.
  }
 \end{proposition}

For data-quantifier free and disjunctive PBESs, we can reformulate the proof graphs as in the following lemma.
\begin{lemma} \label{lem:PG}
 Let $\Epsilon$ be a data-quantifier free and disjunctive PBES, and
 $G = \left<V,{\rightarrow}\right>$ be a graph with $V \subseteq
 \mathrm{sig}(\Epsilon)$ and ${\rightarrow}
 \subseteq V \times V$.
 If $G$ is a proof graph that proves $X(v)=\mathbbold{t}$,
 then there exists a proof graph that proves $X(v)=\mathbbold{t}$ 
 and satisfies all of the following conditions:
 \begin{enumerate}
  \item \label{lem:PG:1} Each vertex has exactly one
    out-going edge from it.
  \item \label{lem:PG:2} For any $(X_i(v), X_j(v'))\in{\to}$,
    there exists $k \in \mathbb{N}$ such that $j=a_{ik}$, $f_{ik}(v)
    = v'$, and $\varphi_{ik}(v) = \mathbbold{t}$.
  \item \label{lem:PG:3} For any infinite sequence $Z_0(v_0)
    \rightarrow Z_1(v_1) \rightarrow \cdots$ along the graph, the
    minimum rank of $Z^{\infty}$ is even, where $Z^{\infty}$ is the
    set of $Z_i$ that occurs infinitely often in the sequence.
 \end{enumerate}
 Conversely, if $G$ satisfies all of these conditions, 
 then it is a proof graph for $\Epsilon$. 
 \end{lemma}
 \begin{proof}
  Let $G$ be a proof graph.
  Then, by Proposition~\ref{prop:PG-loop}, there exists a proof graph $G'=\left<V',{\rightarrow}'\right>$
  that satisfies the condition (\ref{lem:PG:1}).  
  To prove the condition~(\ref{lem:PG:2}) of the lemma,
  assume $(X_i(v),X_j(v'))\in {\to}'$.
  Then, from the condition~(\ref{def:PG:1}) of the proof graph,
  we obtain
  \(
  \interpreted{\varphi_{i}}\theta[\mathbbold{f} / \sig(\Epsilon)][\mathbbold{t} / X_{i}(v)^{\bullet}] \delta [v / d] = \mathbbold{t}
  \)
  for the right-hand side $\varphi_i$ of $X_i$.
  From the definition of data-quantifier free and disjunctive PBESs, it follows that there exists $k$ such that $\varphi_{ik}(v)=\mathbbold{t}$ and $X_{a_{ik}}(f_{ik}(v)) \in X_{i}(v)^{\bullet}$.
  Thus, the condition (\ref{lem:PG:2}) of the lemma holds.
  The condition (\ref{lem:PG:3}) for $G'$ is immediate from the condition~(\ref{def:PG:2}) of the proof graph.

  Next, let $G$ satisfy the conditions of the lemma.
  Consider the condition~(\ref{def:PG:1})
  in the definition of proof graphs for $X_{i}(v) \in V$.
  From the condition (\ref{lem:PG:1}) of the lemma,
  we have an edge $(X_i(v),X_j(v')) \in {\to}$ for some $X_j(v')\in V$.
  From the condition (\ref{lem:PG:2}) of the lemma, there exists $k$ such that $j=a_{ik}$, $f_{ik}(v) = v'$, and $\varphi_{ik}(v) = \mathbbold{t}$.
The condition (\ref{def:PG:2}) of the proof graph follows from the condition (\ref{lem:PG:3}) of the lemma.
  Thus, we can conclude that $G$ is a proof graph.
  \QED
  \end{proof}

Now, we define reduced proof graphs for data-quantifier free and disjunctive
PBESs.
\begin{definition} \label{def:ePG}
For a data-quantifier free and disjunctive PBES $\Epsilon$,
 a directed graph $G = \<V,{\to}>$ with $V \subseteq \mathrm{bnd}(\Epsilon) \times 2^{\mathbb{D}}$
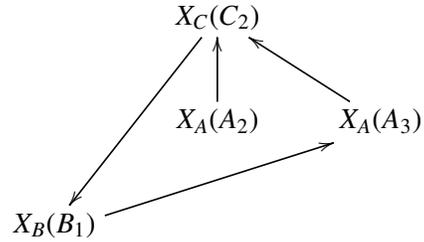
 and ${\to} \subseteq V \times V$ is a \emph{reduced proof graph} if and only if it satisfies all of the following conditions:
 \begin{enumerate}
  \item 
    \label{def:ePG-terminal} Each vertex has exactly one
    out-going edge from it.
  \item \label{def:ePG-succ} For any $(X_i(C), X_j(C'))\in{\to}$,
    there exists $k \in \mathbb{N}$ such that $j=a_{ik}$, $f_{ik}(C)
    \subseteq C'$, and $\varphi_{ik}(v) = \mathbbold{t}$ for any $v
    \in C$.
  \item \label{def:ePG-loop} For any infinite sequence $Z_0(C_0)
    \rightarrow Z_1(C_1) \rightarrow \cdots$ along the graph, the
    minimum rank of $Z^{\infty}$ is even, where $Z^{\infty}$ is the
    set of $Z_i$ that occurs infinitely often in the sequence.
 \end{enumerate}
\end{definition}
We say that a reduced proof graph $G$ \emph{proves} $X_i(v)=\mathbbold{t}$
if and only if there exists some vertex $X_i(C) \in V$ such that $v \in C$.
We can show the relationship between reduced proof graphs and (normal) proof graphs.
 \begin{lemma}\label{lem:ePG2PG}
  For a data-quantifier free and disjunctive PBES and $X(v) \in \sig(\Epsilon)$,
  the existence of a proof graph that proves $X(v)=\mathbbold{t}$
  coincides with the existence of a reduced proof graph that proves $X(v)=\mathbbold{t}$.
 \end{lemma}
 \begin{proof}
  By Lemma~\ref{lem:PG}, there exists a proof graph that satisfies all
  of the conditions in Lemma~\ref{lem:PG}.
  The proof graph is transformed into a reduced one
  by replacing each vertex $X(w)$ with $X(\{w\})$.
  Then, it is trivial
  that the obtained graph is a reduced proof graph that proves $X(v)=\mathbbold{t}$.

  Next, we give a construction of a proof graph $G'$ that proves $X(v)=\mathbbold{t}$ from a given reduced proof graph $G$.
  There exists an infinite path $\pi$ in $G$ starting from $X(C_0)$ such that $v \in C_0$
  from the condition~(\ref{def:ePG-terminal}) in Definition~\ref{def:ePG}.
  Let $\pi$ be the following sequence:
  \[
  \pi:\ X_{\ell_0}(C_0) \rightarrow X_{\ell_1}(C_1) \rightarrow \cdots
  \]
  for some sequence $\ell_0,\ell_1,\ldots$ such that $X_{\ell_0} = X$.
  We construct a sequence
  \[
  \pi':\ X_{\ell_0}(v_0) \rightarrow X_{\ell_1}(v_1) \rightarrow \cdots
  \]
  by choosing $v_m$ from $C_m$ as follows:
  \begin{itemize}
   \item $v_0 = v$
   \item $v_m = f_{\ell_{m-1}k}(v_{m-1})$ for $k$ determined by Definition~\ref{def:ePG}~(\ref{def:ePG-succ}).
  \end{itemize}
  Then, we can regard $\pi'$ as a graph $G'$.
  Since it is easy to show that $G'$ satisfies the conditions in Lemma~\ref{lem:PG}, the obtained graph $G'$ is a proof graph that proves $X(v)=\mathbbold{t}$ by Lemma~\ref{lem:PG}.  
  \QED
  \end{proof}

  From Theorem~\ref{thm:PBES2PG} and Lemma~\ref{lem:ePG2PG}, the membership problem for data-quantifier free and disjunctive PBESs
  is reduced to the problem finding a reduced proof graph.
  Moreover, there exists an instance of the membership problem having a finite reduced proof graph 
  but no finite proof graph as shown in Examples~\ref{eg:PG} and \ref{eg:ePG}.
  In Example~\ref{eg:depsp}, we will show that there exists no finite proof graph for $\Epsilon_2$ by using its dependency space introduced in Section~\ref{sec:construction}.

\section{Dependency Spaces}\label{sec:construction}
In this section, we extend the notion of dependency spaces~\cite{Koolen2015} for reduced proof graphs.
Before proceeding, we recall the notion of congruence on algebra, which we use in this section.

Let $\mathcal{A}=\< A, F^A>$ be a pair such that 
 \begin{itemize}
  \item $A$ is a non-empty set, called \emph{carrier}, and
  \item $F^A$ is a set of partial functions $\alpha^A : A \rightarrow A$.
 \end{itemize}
Then $\mathcal{A}$ is called a \emph{partial algebra}.
An equivalence relation ${\equiv}$ $({}\subseteq A\times A)$
is \emph{congruent}, if the following conditions hold for any $\alpha^{A}\in F^{A}$ and $a, b \in A$ satisfying $a \equiv b$:
  \begin{enumerate}
   \item \label{item:cong1}
   $\alpha^{A}(a)$ is defined if and only if $\alpha^{A}(b)$ is defined, and
   \item \label{item:cong2}
   if $\alpha^{A}(a)$ is defined, then $\alpha^{A}(a) \equiv \alpha^{A}(b)$.
  \end{enumerate}
 The \emph{quotient algebra} of $\mathcal{A}$ with respect to 
 a congruence relation $\equiv$, denoted by $\mathcal{A}/{\equiv}$,
 is the algebra $\mathcal{B}=\< A/{\equiv}, F^B >$, where
 $F^B$ consists of the following functions $\alpha^{B}$ for every $\alpha^A \in F^A$:
		\[
		\alpha^{B}(\eqclass{a}) = \left\{
		\begin{array}{ll}
		 \eqclass{\alpha^A(a)}, & \text{if $\alpha^A(a)$ is defined}\\
		 \text{undefined,} & \text{otherwise}
		\end{array}
		\right.
		\]
 Note that $\eqclass{a}$ denotes the equivalence class containing $a$.

 A reduced dependency space includes at least one reduced proof graph if it exists.
Then, we can construct a reduced proof graph by deleting vertices and edges from a reduced dependency space so that
it satisfies all of the conditions of Definition~\ref{def:ePG}.

\begin{definition}\label{def:indA}
 For a partial algebra $\mathcal{A}=\< A, F^A >$, the \emph{graph induced from} $\cal{A}$ is defined as
 the directed graph $\< A, {\rightarrow}>$, where
 \[
 {\rightarrow}=
 \{(u,\alpha^{A}(u)) \mid u\in A,\; \alpha^{A}\in F^A,\; \alpha^{A}(u) \text{ is defined} \}.
 \]
\end{definition}
A congruence on the dependency space for a given PBES determines a reduced dependency space.
\begin{definition}
The \emph{dependency space} of 
a given data-quantifier free and disjunctive PBES $\Epsilon$ (we use the notation of Definition~\ref{def:dq-free-dis}), is the graph induced from the following partial algebra
$\mathcal{A} = \< A, F^A>$, where
\begin{itemize}
  \item $A = \{ X_i(v) \mid v\in \mathbb{D},\; i\in\{1,\ldots,n\} \}$, and
  \item $F^A = \{ \alpha^{A}_{ik} \mid i\in\{1,\ldots,n\},\; k\in\{1,\ldots,m_i\}\}$, where
    \[
    \alpha_{ik}^{A}(X_j(v)) = \left\{
    \begin{array}{ll}
      X_{a_{ik}}(f_{ik}(v)), & \text{if $i=j$ and $\varphi_{ik}(v) = \mathbbold{t}$}\\
      \text{undefined}, & \text{otherwise}\\
    \end{array}
    \right.
    \]
\end{itemize}

Moreover, the graph induced from $\mathcal{A}/{\equiv}$ for a congruence relation ${\equiv}$ with respect to $\mathcal{A}$ is a \emph{reduced dependency space} of PBES $\Epsilon$.
\end{definition}
Note that the equivalence classes are not always finite.

 \begin{example}\label{eg:depsp}
  The algebra $\mathcal{A}$ for the PBES $\Epsilon_2$ in Example~\ref{eg:PBES} is $\<A,\{\alpha^A_{11},\alpha^A_{12},\alpha^A_{21},\alpha^A_{22}\}>$, where
  \[\begin{array}{l}
     A = \{X_1(v),X_2(v) \mid v\in\mathbb{N}\},\\
     \alpha_{11}(X_1(v)) = X_1(v+1) \text{\quad for any $v\in\mathbb{N}$},\\
     \alpha_{12}(X_1(v)) = X_2(v) \text{\quad if $v \geq 1$},\\
     \alpha_{21}(X_2(v)) = X_2(v+1) \text{\quad for any $v\in\mathbb{N}$},\\
     \alpha_{22}(X_2(v)) = X_1(v) \text{\quad if $v = 0$}.
  \end{array}\]
  This is illustrated as follows:
 \[
 \xymatrix{
 X_1(0) \ar[r]^(0.5){\alpha_{11}}
       & X_1(1) \ar[r]^(0.5){\alpha_{11}} \ar[d]^{\alpha_{12}}
       & X_1(2) \ar[r]^(0.5){\alpha_{11}} \ar[d]^{\alpha_{12}}
       & \cdots\\
 X_2(0) \ar[u]^{\alpha_{22}} \ar[r]^(0.5){\alpha_{21}}
       & X_2(1) \ar[r]^(0.5){\alpha_{21}}
       & X_2(2) \ar[r]^(0.5){\alpha_{21}}
       & \cdots
 }
 \]
 The dependency space induced from $\mathcal{A}$ is the graph obtained from the above graph
 by removing function symbols on the edges.

  Remark that the dependency space contains every proof graph as a sub-graph for a disjunctive PBES~\cite{Koolen2015}.
  Because a proof graph must have exactly one out-going edge, it is trivial that there exists no finite proof graph for $\Epsilon_2$.
  
  Let ${\equiv}$ be a congruence relation described below.
  \[\begin{array}{l}
   X_1(v) \equiv X_1(w) \text{\quad for $v$ and $w$ such that $v \geq 1 \wedge w \geq 1 $}\\
   X_2(v) \equiv X_2(w) \text{\quad for $v$ and $w$ such that $v \geq 1 \wedge w \geq 1$}
  \end{array}\]
  Then, the carrier $A/{\equiv}$ of the quotient algebra $\mathcal{A}/{\equiv}$ is
  $\{\{X_1(0)\}, \{X_1(v)\mid v\geq 1\},
    \{X_2(0)\}, \{X_2(v)\mid v\geq 1\}\}$.
 The following graph is the reduced dependency space induced from the quotient algebra:
 \[
 \xymatrix{
 \{X_1(0)\} \ar[r]^(0.4){\alpha_{11}}  
       & \{X_1(v)\mid v\geq 1\} \ar[d]^{\alpha_{12}} \ar@(ul,ur)^{\alpha_{11}} \\
 \{X_2(0)\} \ar[u]^{\alpha_{22}} \ar[r]^(0.4){\alpha_{21}}
       & \{X_2(v)\mid v\geq 1\} \ar@(ld,rd)_{\alpha_{21}}
 }
 \]
Note that the vertices are also written as 
$X_1(\{0\}), X_1(\{1,2,\ldots\}), X_2(\{0\}), X_2(\{1,2,\ldots\})$,
respectively.
From this dependency space, we can easily extract the reduced proof graph in Example~\ref{eg:ePG}.
\end{example}
  Hereafter, we use the above notation to describe a reduced dependency space.
 \begin{example}
  Consider the data-quantifier free and disjunctive PBES $\Epsilon_3$ given as follows:
  \[
  \nu X(d: N) = \left(d \bmod 3 < 2 \wedge X\left(d+1\right)\right)\; \vee\; \left(d \bmod 3 = 1\ \wedge\ X\left(d+2\right)\right)
  \]
  The following graph is a reduced dependency space of $\Epsilon_3$:
  \[
  \xymatrix{
  X(N_0) \ar@<-0.5ex>[r] & X(N_1) \ar@<-0.5ex>[l] \ar[r] & X(N_2) \\
  }
  \]
  where $N_i = \{n \mid n \bmod 3 = i\}$ for each $i \in \{0,1,2\}$.
  This reduced dependency space includes a reduced proof graph shown below:
  \[
  \xymatrix{
  X(N_0) \ar@<-0.5ex>[r] & X(N_1) \ar@<-0.5ex>[l]
  }
  \]
  We can see $\interpreted{\Epsilon_3}(X)(d) = \mathbbold{t}$ iff
  $d \in N_0 \cup N_1$ from the reduced proof graph.
 \end{example}

We show a property of dependency spaces.
\begin{lemma}\label{lem:depsp}
 Let $S$ be a reduced dependency space for a data-quantifier free and disjunctive PBES $\Epsilon$.
 If there exists
 a proof graph that proves $X(v) = \mathbbold{t}$ for $\Epsilon$, 
 then there exists a sub-graph of $S$ that is a reduced proof graph proving $X(v) = \mathbbold{t}$.
\end{lemma}
\begin{proof}
 We give a way to construct a sub-graph of $S$ from a
 given proof graph $G$.
 By Proposition~\ref{prop:PG-loop},
 $G$ consists of an infinite path $\pi$ starting from $X(v)$.
  Let $\pi$ be the following sequence:
  \[
  \pi:\ X(v) = X_{\ell_0}(v_0) \rightarrow X_{\ell_1}(v_1) \rightarrow \cdots
  \]
 for some sequence $\ell_0,\ell_1,\ldots$.
 Let $G'$ be the graph consisting of the following sequence $\pi'$
  \[
  \pi':\ \eqclass{X_{\ell_0}(v_0)} \rightarrow \eqclass{X_{\ell_1}(v_1)} \rightarrow \cdots
  \]
where ${\equiv}$ is the congruence relation that characterizes $S$.

 First, we show that $G'$ is a sub-graph of $S$.
 Obviously, all vertices in $G'$ are also in $S$.
 Let $m\geq 0$.
 Since $X_{\ell_m}(v_m) \to X_{\ell_{m+1}}(v_{m+1})$ appears in the proof graph $G$,
 there exists $k \in \mathbb{N}$ such that $a_{\ell_mk} = \ell_{m+1}, f_{\ell_mk}(v_m) = v_{m+1}$ and $\varphi_{\ell_{m}k}(v_m) = \mathbbold{t}$
 from Lemma~\ref{lem:PG}.
 From the definition of congruence relations, $\alpha^A_{\ell_mk}(X_{\ell_m}(v_m))$ is defined and its value is $X_{\ell_{m+1}}(v_{m+1})$.
 Thus, we have $\eqclass{\alpha^A_{\ell_mk}(X_{\ell_m}(v_m))} = \eqclass{X_{\ell_{m+1}}(v_{m+1})}$, and $\eqclass{X_{\ell_m}(v_m)} \to \eqclass{X_{\ell_{m+1}}(v_{m+1})}$ also appears in $S$.
 
 Next, we show that $G'$ is a reduced proof graph.
 The conditions (\ref{def:ePG-terminal}) and (\ref{def:ePG-loop}) in Definition \ref{def:ePG}
 hold immediately from the form of $\pi'$ and the condition (\ref{lem:PG:1}) in Lemma \ref{lem:PG}.
  Since ${\equiv}$ is congruent and $\alpha^A_{\ell_mk}(X_{\ell_m}(v_m))$ is defined, it follows that
  $X_{\ell_{m+1}}(f_{\ell_m k}(v)) \in \eqclass{X_{\ell_{m+1}}(v_{m+1})}$ for any $X_{\ell_m}(v) \in \eqclass{X_{\ell_m}(v_m)}$.
 Therefore, the condition (\ref{def:ePG-succ}) holds.
 \QED
\end{proof}

From Lemma~\ref{lem:depsp}, if a proof graph exists then 
there exists a reduced proof graph as a sub-graph of the dependency space.
For example, we see that the reduced proof graph in Example~\ref{eg:ePG} is a sub-graph of the dependency space shown in Example~\ref{eg:depsp}.

Note that data-quantifier free and conjunctive PBESs can be defined dually to data-quantifier free and disjunctive PBESs,
and we have the dual results for data-quantifier free and conjunctive PBESs.

\section{Graph Construction}\label{sec:semi-alg}
In this section, we propose a procedure to construct the reduced
dependency space induced from the maximal congruence,
where the maximal congruence induces the most general reduced dependency space.
We start from $n$
vertices $\{X_1(v)\mid v\in\mathbb{D}\}, \ldots, \{X_n(v)\mid
v\in\mathbb{D}\}$ and divide the sets until the conditions of the
congruence relation are satisfied.
This procedure is captured as repetition of division operations on a partition of
$\mathbb{D}$ for each $i\in\{1,\ldots,n\}$, where a family $\Phi$ of
sets is a \emph{partition} of $\mathbb{D}$ if every two different sets in $\Phi$ are
disjoint and the union of $\Phi$ is equal to $\mathbb{D}$.
At the end of this section, we prove soundness and completeness of the procedure, i.e.,
the procedure returns the most general reduced dependency space if it is finite.

We define a function $H$ that takes a tuple of partitions $\<
\Psi_1,\ldots,\Psi_n >$ and returns a tuple of partitions
obtained by doing necessary division operations to elements $\Psi_i$'s.  The procedure
repeatedly applies $F$ to the initial tuple of partitions until it
saturates.  If it halts, the resulted tuple induces a reduced dependency 
space.
In the procedure, Boolean expressions are used to denote (possibly infinite) subsets
of the data domain $\mathbb{D}$.
In other words, a Boolean expression $\phi(d)$ can be regarded as a set $\{v \in \mathbb{D} \mid \phi(v)\}$. 
In the sequel, we abuse operations on sets to denote Boolean operations.
For example, we may use the binary operators $\cap$ (resp.\ $\subseteq$) on sets
for intersection (resp.\ implication) in Boolean expressions.

We give intuitive explanation of the division.  Suppose a formula $\varphi_{ik}(d)$ is $d < 10$ in a given PBES.
Then, we have to divide the data domain $\mathbb{D}$ into $\{v \in \mathbb{D} \mid v < 10\}$ and $\{v \in \mathbb{D} \mid v \not < 10\}$, because
the condition (\ref{item:cong1}) of the congruence relation requires the coincidence of the defined-ness of $\alpha_{ik}$ for all data in a set,
where $\alpha_{ik}(v)$ is defined if and only if $\varphi_{ik}(v)$ holds. 
The condition (\ref{item:cong2}) requests a similar division.
Now we prepare this operation.
In general, we must divide each set in a partition $\Phi$ according to a formula $\psi$.
We define this division operation as follows:
\[
 \Phi \otimes \psi := \{\phi \cap \psi \mid \phi \in \Phi\} \cup
 \{\phi \cap \ov{\psi} \mid \phi \in \Phi\}
\]
This operator obviously satisfies $(\Phi \otimes \psi_1) \otimes \psi_2 = (\Phi \otimes \psi_2) \otimes \psi_1$, thus
we can naturally extend it on sets of formulas as follows:
\[
 \Phi \otimes \{\psi_1,\ldots,\psi_p\} = \Phi \otimes \psi_1 \otimes \cdots \otimes \psi_p
\]
It is easily shown that if $\Phi$ is a partition of $\mathbb{D}$, then $\Phi \otimes \Psi'$ is also a partition of $\mathbb{D}$ for a set $\Psi'$ of formulas.

In constructing partitions of data sets, it is not necessary to apply the division due to the condition (\ref{item:cong1}) for the congruence more than once.  Thus we use partitions
resulted by such a division as the initial ones.
The \emph{tuple of initial partitions} are $\<\Omega_1,\dots,\Omega_n>$,
where
\(
 \Omega_i = \{ \mathbb{D} \} \otimes \{\varphi_{i1},\ldots, \varphi_{im_i}\}
\).
Note that $\Omega_i$ consists of at most $2^{m_i}$ sets, because each element $\omega$
is included in $\varphi_{ik}$ or $\ov{\varphi_{ik}}$.

The condition (\ref{item:cong2}) for the congruence requires that
$X_i(v) \equiv X_i(w) \implies \alpha_{ik}^{A}(X_i(v)) \equiv \alpha_{ik}^{A}(X_i(w))$
if $\alpha^A_{ik}(X_i(v))$ is defined.
We recall this condition by an example.
We assume the current partitions $\<\Phi_1,\Phi_2> = \langle\{d\leq 0,\; d>0\},\; \{d\leq 0,\; d > 0\}\rangle$
and a clause $d > 0 \wedge X_2(d-1)$ in the equation for $X_1$.
The set represented by $d>0$ in $\Phi_1$ obviously satisfies the formula $d>0$ in the clause, thus all elements in $\{d-1 \mid d>0\}$ should be included in a set in $\Phi_2$, but they are not included in.
Therefore, we will divide the set represented by $d>0$ in $\Phi_1$ into two sets as illustrated by $\{d>0\} \otimes \{d-1\leq 0,\; d-1>0\} = \{d>0\wedge d\leq 1,\; d>1\}$.
This is formalized as follows.
\begin{definition}
 The \emph{partition function} $H_{ik}$ for each $i$ and $k$ is defined as follows:
 \[
 \begin{array}{l}
  H_{ik}(\<\Psi_1,\dots,\Psi_n>) := \<\Psi'_1,\dots,\Psi'_n>\\
  \Psi'_j = \left\{\begin{array}{ll}
			 \Psi_j & (i\ne j)\\
					\{\psi \in \Psi_j \mid \psi \subseteq \ov{\varphi_{ik}}\} \cup
					 \left(\{\psi \in \Psi_j \mid \psi \subseteq \varphi_{ik}\} \otimes \Psi_{a_{ik}}[f_{ik}(d)/d]\right)
					 & (i=j)
				   \end{array}
		   \right.
 \end{array}
 \]
 where $\Psi[d'/d]$ is the set of formulas each of which is obtained from a formula in $\Psi$ by replacing $d$ with $d'$.
\end{definition}
Here $\psi$ satisfying $\psi \subseteq \ov{\varphi_{ik}}$ is not divided, because $\alpha^A_{ik}(X_i(v))$ is not defined for $v$ in the set represented by $\psi$.  On the other hand, $\psi$ satisfying $\psi \subseteq \varphi_{ik}$ is divided so that the image $f_{ik}(\psi)$ is included in some set in $\Psi_{a_{ik}}$, because $\alpha^A_{ik}(X_i(v))$ is defined.

 Partition functions are bundled as follows:
 \[
 \begin{array}{rcl}
  H(\<\Psi_1,\dots,\Psi_n>) &:=& \left(H_1\circ\dots\circ H_n\right)(\<\Psi_1,\dots,\Psi_n>)\\
   H_i(\<\Psi_1,\dots,\Psi_n>) &:=& \left(H_{i1}\circ\dots\circ H_{i{m_i}}\right)(\<\Psi_1,\dots,\Psi_n>)
 \end{array}
 \]
 where $\circ$ denotes composition, i.e., $(f\circ g)(x)=g(f(x))$.
 The function $H_i$ denotes the partition of $\Psi_i$ using some functions $H_{i1},\dots,H_{i{m_i}}$.
 The function $H$ takes a series of partitionings due to the condition (2) for the congruence. 

 We define the partition procedure that applies the partition function $H$ to initial partitions $\<\Omega_1,\ldots,\Omega_n>$ until it saturates.
 We write the family of the partitions obtained from the procedure as $H^{\infty}(\<\Omega_1,\ldots,\Omega_n>)$.

\begin{example}
 Consider the data-quantifier free and disjunctive PBES $\Epsilon_4$ given as follows:
  \[
  \begin{array}[t]{rcl}
   \nu X_1(d: N)& = & (\mathrm{true} \wedge X_1(d+1))\ \vee\ \left(d \geq 1\ \wedge\ X_2\left(d-1\right)\right)\\
   \mu X_2(d: N)& = & (\mathrm{true} \wedge X_2(d+1))\ \vee\ \left(d \leq 0\ \wedge\ X_1\left(d\right)\right)
  \end{array}
  \]
 The reduced dependency space for $\Epsilon_4$ is:

 
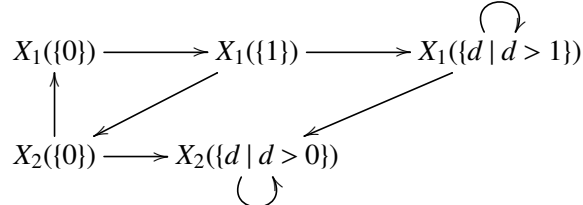
\[
\smallskip
 \xymatrix{
 X_1(\{0\}) \ar[r] & X_1(\{1\}) \ar[r] \ar[ld]& X_1(\{d \mid d > 1\}) \ar@(ul,ur) \ar[ld]& \\
 X_2(\{0\}) \ar[u]\ar[r] & X_2(\{d \mid d > 0\}) \ar@(ld,rd) &
 }
\smallskip
 \]

\noindent
 In order to construct this, we first calculate initial partitions.
 \[
 \begin{array}[t]{rcccccl}
   \Omega_1 &=& (\{\mathrm{true}\} \otimes \mathrm{true}) \otimes (d \geq 1) &=&
   \{\mathrm{true}\} \otimes (d \geq 1) &=& \{d \geq 1, d<1\}\\
   \Omega_2 &=& (\{\mathrm{true}\} \otimes \mathrm{true}) \otimes (d \leq 0) &=&
   \{\mathrm{true}\} \otimes (d \leq 0) &=& \{d \leq 0, d>0\}
 \end{array}
 \]
We omit the element 
equivalent to $\mathrm{false}$
from partitions because it represents an empty set.

 Next, we apply $H$ to $\<\Omega_1,\Omega_2>$.
 \[
 \begin{array}{rcl}
  (H_{11} \circ H_{12})(\<\Omega_1,\Omega_2>) &=& H_{12}(\<\emptyset \cup (\{d \geq 1, d < 1\} \otimes \{d+1 \geq 1, d+1<1\}), \Omega_2\bigr>)\\
  &=& H_{12}(\<\{d \geq 1, d < 1\}, \Omega_2 >)\\
  &=& \<{\{d<1\} \cup (\{d \geq 1\} \otimes \{d-1 \leq 0, d-1 > 0\}), \Omega_2}>\\
  &=& \<{\{d<1, d=1, d>1\}, \Omega_2}>\\
 \end{array}
 \]
 We also apply $H_{21}$ and $H_{22}$ in a similar way, and $\Omega_2$ does not change.
 As a result, we obtain $H(\<\Omega_1,\Omega_2>) = \langle\{d<1, d=1, d>1\}, \Omega_2\rangle$, which is already a fixed point.
 Hence, the procedure returns
 \[
  H^{\infty}\left(\<\Omega_1,\Omega_2>\right) = 
              \langle\{d<1, d=1, d>1\},
		\{d\leq 0, d>0\}\rangle
 \]
 This partition induces the set of vertices in the reduced dependency space.
\end{example}

We prepare some technical lemmas on the operation $H$.
\begin{lemma}\label{lem:cong}
 For a tuple $\Omega = \<\Omega_1,\ldots,\Omega_n>$ of initial partitions,
 the tuple $H^{\infty}(\Omega)$ of the partitions is the quotient set of the algebra 
 induced from a PBES $\Epsilon$ for some congruence.
 In other words, letting $H^{\infty}(\Omega)$ = $\<\Omega'_1,\ldots,\Omega'_n>$, the following two properties hold:
 \[\begin{array}{ll}
 \forall i, \forall k, \forall \omega \in \Omega'_i, (\omega \subseteq \varphi_{ik}) \vee (\omega \subseteq \ov{\varphi_{ik}}),\\
 \forall i, \forall k, \forall \omega \in \Omega'_i, \forall \omega' \in \Omega'_{a_{ik}},
 (\omega \subseteq \varphi_{ik}) \Rightarrow ((\omega \subseteq \omega'[f_{ik}(d)/d])
 \vee (\omega \subseteq \ov{\omega'[f_{ik}(d)/d]}))
 \end{array}\]
\end{lemma}
\begin{proof}
 The former property for $\Omega_i$'s follows from the definition of initial partitions and the fact that  $H$ preserves the property.
 We prove the latter property by contradiction.
 Let 
 $\omega \subseteq \varphi_{ik}$,
 $\omega \not \subseteq \omega'[f_{ik}(d)/d]$, and $\omega \not \subseteq \ov{\omega'[f_{ik}(d)/d]}$ for some $i,k,\omega \in \Omega'_i$, and $\omega' \in \Omega'_{a_{ik}}$.
 Then, we have $\omega\cap\omega'' \neq \emptyset$ and $\omega\cap\ \ov{\omega''} \neq \emptyset$, where $\omega''$ denotes $\omega'[f_{ik}(d)/d]$.
 This implies that $\{\omega\} \otimes \omega''$ results in two non-empty sets
 $\omega \cap \omega''$ and $\omega \cap \ov{\omega''}$ by division of $\omega$.
 Combining this and the fact that $\Omega'_1$ is a partition, it follows that $\<\Omega'_1,\ldots,\Omega'_n>$ is not a fixed point of $H_{ik}$, which contradicts the assumption.
 \QED
\end{proof}

For a tuple of partitions $\Psi = \<\Psi_1, \ldots , \Psi_n>$ on a data-quantifier free and disjunctive PBES $\Epsilon$,
we define a relation ${\sim_\Psi}$ on the partial algebra $\<A,F^A>$ defined
by $\Epsilon$ (see Definition~\ref{def:indA}) as follows:
\[
  X_i(v) \sim_\Psi X_j(v') \iff i=j \wedge \exists \psi\in \Psi_j\ (\psi(v)\wedge\psi(v'))
\]
  Note that it is trivial that $\sim_{\Psi}$ is an equivalence relation since each $\Psi_i$ is a partition of $\mathbb{D}$.

\begin{lemma}
Let $\<\Omega'_1,\ldots,\Omega'_n>$ be a fixed point of $H$.
Then $\sim_{\Omega'}$ is congruent.
\end{lemma}
\begin{proof}
  Let $X_j(v) \sim_{\Omega'} X_j(v')$.  Then there exists $\omega\in \Omega'_j$ such that $\omega(v)$ and $\omega(v')$ hold.
  Suppose $\alpha^A_{ik}(X_j(v))$ is defined, then $i=j$ and $\varphi_{ik}(v)$ hold.
  From the former property of Lemma~\ref{lem:cong}, $\varphi_{ik}(v')$ holds.  Thus, $\alpha^A_{ik}(X_j(v'))$ is also defined, which shows (\ref{item:cong1}) of the definition of congruence.

If $\alpha^A_{ik}(X_j(v))$ is defined, it is equal to $X_{a_{ik}}(f_{ik}(v))$ and
also $\alpha^A_{ik}(X_j(v'))=X_{a_{ik}}(f_{ik}(v'))$.
Since $\Omega'_{a_{ik}}$ is a partition, there exists $\omega'\in \Omega'_{a_{ik}}$ such that $\omega'(f_{ik}(v))$ holds.
From the second property of  Lemma~\ref{lem:cong}, 
$\omega'(f_{ik}(v'))$ also holds, which shows (\ref{item:cong2}) of the definition of congruence.
\QED
\end{proof}

The following theorem follows from this lemma.
\begin{theorem}
 If the procedure terminates, then the partitions $H^{\infty}(\<\Omega_1,\ldots,\Omega_n>)$ are the vertices of a reduced dependency space.
\end{theorem}

We prepare lemmas for proving the completeness of $H^\infty(\<\Omega_1,\ldots,\Omega_n>)$.
\begin{lemma}
 Let ${\equiv}$ be a congruence on a given data-quantifier free and disjunctive PBES,
 and $\Omega = \<\Omega_1,\ldots,\Omega_n>$ be the tuple of initial partitions.
 Then, ${\sim_{\Omega}} \supseteq {\equiv}$.
\end{lemma}
\begin{proof}
 We show the lemma by contradiction.
 Suppose there exist $X_i(v)$ and $X_i(w)$ such that $X_i(v) \equiv X_i(w)$ and $X_i(v) \not\sim_{\Omega} X_i(w)$
 for some $v,w \in \mathbb{D}$ and $i \in \{1,\ldots,n\}$.
 Since $\Omega_i$ is a partition, there exists a $\omega\in\Omega_i$ such that $\omega(v)$ holds.
 Because $X_i(v) \not\sim_{\Omega} X_i(w)$, $\omega(w)$ does not hold.
 This means from the definition of initial partition that $\phi_{ik}(v)$ holds but $\phi_{ik}(w)$ does not for some $k \in\{1,\ldots,m_i\}$.
 Thus, $\alpha^A_{ik}(X_i(v))$ is defined but $\alpha^A_{ik}(X_i(v))$ is not defined, which contradicts $X_i(v) \equiv X_i(w)$.
 \QED
\end{proof}
\begin{lemma}
 Let ${\equiv}$ be a congruence on $\Epsilon$, and $\Psi$ be a tuple of partitions.
 Then, ${\sim_{\Psi}} \supseteq {\equiv}$ implies ${\sim_{H(\Psi)}} \supseteq {\equiv}$.
\end{lemma}
\begin{proof}
 From the definition of $H$, it is enough to show that ${\sim_{\Psi}} \supseteq {\equiv}$ implies ${\sim_{H_{ik}(\Psi)}} \supseteq {\equiv}$ for arbitrary $1\leq i\le n$ and $1\leq k\leq m_i$.
 We show this by contradiction.
 We assume ${\sim_{\Psi}} \supseteq {\equiv}$ and ${\sim_{H_{ik}(\Psi)}} \not\supseteq {\equiv}$.  
 Let $\Psi = \<\Psi_1,\ldots,\Psi_n>$ and $H_{ik}(\Psi) = \<\Psi'_1,\ldots,\Psi'_n>$.
 Then, from the definition of $H_{ik}$, we have ${\Psi'}_j=\Psi_j$ for any $j$ ($\not=i$).
 This implies that $X_i(v) \equiv X_i(w)$ and $X_i(v) \not\sim_{H_{ik}(\Psi)} X_i(w)$ for some $v,w\in\mathbb{D}$.
Note that
\[
   \Psi'_i=\{\psi \in \Psi_i \mid \psi \subseteq \ov{\varphi_{ik}}\} \cup
   \left(\{\psi \in \Psi_i \mid \psi \subseteq \varphi_{ik}\} \otimes \Psi_{a_{ik}}[f_{ik}(d)/d]\right).
\]

 From $X_i(v) \equiv X_i(w)$ and ${\sim_\Psi}\supseteq{\equiv}$, there exists $\psi\in\Psi_i$ such that $\psi(v)$ and $\psi(w)$ hold.  
 Since $\psi\not\in \Psi'_i$ due to $X_i(v) \not\sim_{H_{ik}(\Psi)} X_i(w)$, the formula $\psi$ is divided by a formula $\omega[f_{ik}(d)/d]$ for some $\omega\in \Psi_{a_{ik}}$.
 Thus, $\omega[f_{ik}(d)/d](v)$ holds but $\omega[f_{ik}(d)/d](w)$ does not without loss of generality, and $\varphi_{ik}(v)$ and $\varphi_{ik}(w)$ also hold.
 The former means that $\omega(f_{ik}(v))$ holds but $\omega(f_{ik}(w))$ does not.
 Since $\alpha_{ik}(X_i(v))=X_{a_{ik}}(f_{ik}(v))$ and $\alpha_{ik}(X_i(w))=X_{a_{ik}}(f_{ik}(w))$, we obtain $\alpha_{ik}(X_i(v)) \not\sim_{\Psi} \alpha_{ik}(X_i(w))$.
 Since ${\sim_\Psi} \supseteq {\equiv}$, we get $\alpha_{ik}(X_i(v)) \not\equiv \alpha_{ik}(X_i(w))$,
 which contradicts the assumption $X_i(v) \equiv X_i(w)$.
\QED
\end{proof}

The completeness follows from these lemmas.
\begin{theorem}
 Let $\Omega'$ be the least fixed point of $H$ containing the tuple $\Omega$ of initial partitions.  Then $\sim_{\Omega'}$ is the maximal congruence.
 Thus, $H^\infty(\Omega)$ induces the most general reduced dependency space.
\end{theorem}

The following corollary can be immediately obtained from the above theorem.
\begin{corollary}
 For a given PBES, $H^\infty(\Omega)$ induces a finite reduced dependency space, if it exists.
\end{corollary}
This corollary says that a finite reduced dependency space is eventually found by $H^\infty(\<\Omega_1,\ldots,\Omega_n>)$ if it exists.
There exists, however, a data-quantifier free and disjunctive PBES having a finite reduced proof graph
but no finite reduced dependency space.
This is shown by the following example.
\begin{example}
 Consider the following data-quantifier free and disjunctive PBES:
 \[
 \begin{array}[t]{rcl}
  \nu X_1(d: N)& = & (\mathrm{true} \wedge X_1(d+1))\\
  \mu X_2(d: N)& = & (d > 0 \wedge X_2(d-1))\vee\left(\mathrm{true} \wedge X_1(d)\right).
 \end{array}
 \]
 There is an infinite proof graph of $X_2(0)$ as shown below, but no finite one.
 
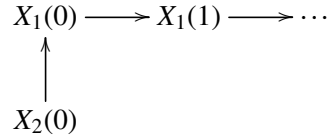
\[
 \xymatrix{
 X_1(0) \ar[r]
 & X_1(1) \ar[r]
 & \cdots \\
 X_2(0) \ar[u]
 }
 \]
 On the other hand, there is a finite reduced proof graph of $X_2(0)$ shown as follows:
 
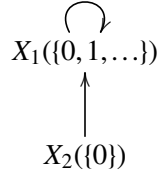
\vspace{1em}
 \[
 \xymatrix{
 X_1(\{0,1,\ldots\}) \ar@(ul,ur) \\
 X_2(\{0\}) \ar[u]
 }
 \]
 There exists, however no finite reduced dependency space, because $H^\infty(\Omega)$ induces the following infinite reduced dependency space.
 
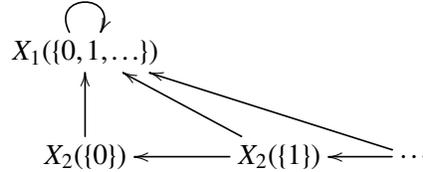
\vspace{1em}
 \[
 \xymatrix{
 X_1(\{0,1,\ldots\}) \ar@(ul,ur) \\
 X_2(\{0\}) \ar[u]
 & X_2(\{1\}) \ar[lu] \ar[l]
 & \ar[llu] \ar[l] \cdots
 }
 \]
\end{example}
 This example shows that a reduced dependency space may be possibly infinite although a finite reduced proof graph exists.

Considering an implementation of this procedure, it is reasonable to use a set of Boolean expressions for representing a partition.
The division operation $\otimes$ in the procedure may produce 
unsatisfiable expressions, which is unnecessary in partitions and hence should be removed.
An incomplete unsatisfiability check easily causes a non-termination of the procedure for a PBES, even if the procedure with complete satisfiability check terminates.
Thus, the unsatisfiability check of Boolean expressions is one of the most important issues in implementing the procedure.

For instance, examples illustrated in this paper are all in the class of Presburger arithmetic, which is the first-order theory of the natural numbers which has addition.
It is known that the unsatisfiability check of Boolean expressions in this class is decidable~\cite{presburger1931vollständigkeit}.
Therefore, the procedure enjoys the completeness property for this class.

\section{An Example: Trading Problem}
This section illustrates a simple but more realistic application of our method.

There are villages $A,B$ and they trade with each other.
The trade is taken by using a truck initially located in $A$.
Whenever the truck moves, each village earns profit according to the moving path of the truck, and requires fixed cost for a living.
Moreover, there is a place $C$ which supplies the truck with fuel.
We consider the following problem: \emph{Is there a schedule for the truck satisfying that
the balance of each village is always positive and the truck visits place $C$ infinitely often?}

The following graph shows the restriction of the truck movement, 
where the pair $(x,y)$ on each arrow denotes the amount of money which $A$ and $B$ get by trading, respectively.
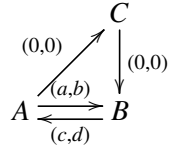
\[
 \xymatrix{
 &C \ar[d]^{(0,0)}\\
 A \ar[ru]^{(0,0)} \ar@<0.5ex>[r]^{(a,b)} & B \ar@<0.5ex>[l]^{(c,d)}
 }
\]
We write the living expenses of $A$ (resp.\ $B$) as $E_A$ (resp.\ $E_B$).

This problem is encoded as a PBES in the following way.
The PBES has three predicate variables $X_A$, $X_B$, and $X_C$, where $X_\alpha(x,y)$ is true if and only if there is a successful track schedule from the configuration, where the amounts of money in $A$ and $B$ are $x$ and $y$, respectively, and the truck is located in $\alpha$.
\[
\begin{array}[t]{rcl}
 \nu X_C(x,y)& = & x-E_A \geq 0 \wedge y-E_B \geq 0 \wedge X_B(x-E_A, y-E_B)\\
 \mu X_A(x,y)& = & (x-E_A \geq 0 \wedge y-E_B \geq 0 \wedge X_C(x-E_A, y-E_B))\\
 && \vee (x+a-E_A \geq 0 \wedge y+b-E_B \geq 0 \wedge X_B(x+a-E_A, y+b-E_B))\\
 \mu X_B(x,y)& = & x+c-E_A \geq 0 \wedge y+d-E_B \geq 0 \wedge X_A(x+c-E_A, y+d-E_B)
\end{array}
\]

Let $(a,b,c,d) = (4,3,3,4)$ and $E_A = E_B = 1$, then the PBES can be simplified as below:
\[
\begin{array}[t]{rcl}
 \nu X_C(x,y)& = & x \geq 1 \wedge y \geq 1 \wedge X_B(x-1, y-1)\\
 \mu X_A(x,y)& = & (x \geq 1 \wedge y \geq 1 \wedge X_B(x-1, y-1))
  \vee (\text{true} \wedge X_B(x+3, y+2))\\
 \mu X_B(x,y)& = & \text{true} \wedge X_A(x+2, y+3)
\end{array}
\]
For this problem, our procedure produces the partitions $\<\{C_1,C_2\}, \{A_1, A_2, A_3\}, \{B_1\}>$, where
\[
 \begin{array}{lll}
 C_1=x\geq 1 \wedge y \geq 1,       & C_2= \lnot(x\geq 1 \wedge y \geq 1), &\\
 A_1=\lnot(x\geq 1 \wedge y \geq 1),& A_2= (x\geq 1 \wedge y \geq 1)\wedge \lnot(x\geq 2 \wedge y \geq 2), & A_3= x\geq 2 \wedge y \geq 2, \text{ and}\\
 B_1=\mathrm{true}.
 \end{array}
\]
The reduced dependency space induced from these partitions is illustrated in Figure~\ref{fig:depsp}.

We show a reduced proof graph in Figure~\ref{fig:rpg}.
Since there is no reduced proof graph of $X_C(x,y)$ satisfying $C_1$ nor $X_A(x,y)$ satisfying $A_1$, the displayed reduced proof graph characterizes initial configuration having successful track schedules. 
Considering the case that the initial location of the truck is in $A$, 
the condition for the initial amount $x$ for $A$ and $y$ for $B$ is $A_2 \vee A_3$, which is simplified as $x \geq 1 \wedge y \geq 1$.
\begin{figure}[ht]
 \begin{minipage}{0.5\hsize}
  \[
  \xymatrix{
  X_C(C_1) & X_C(C_2)\ar[ldd] \\
  X_A(A_1) & X_A(A_2)\ar[u]\ar[ld] & X_A(A_3)\ar[lu]\ar@<-0.5ex>[lld]\\
  X_B(B_1) \ar@<-0.5ex>[rru]
  }
  \]
  \caption{reduced dependency space}
  \label{fig:depsp}
 \end{minipage}
 \begin{minipage}{0.5\hsize}
  \[
  \xymatrix{
   & X_C(C_2)\ar[ldd] \\
   & X_A(A_2)\ar[u] & X_A(A_3)\ar[lu]\\
  X_B(B_1) \ar@<-0.5ex>[rru]
  }
  \]
  \caption{reduced proof graph}
  \label{fig:rpg}
 \end{minipage}
\end{figure}

\section{Related Work}
Approaches to transforming an infinite domain (or state space) into an equivalent finite domain (or state space)
with regard to a certain criterion such as behavioral equivalence or congruence with operations
can be found in various topics in logics and formal verification.

First, of course, the minimization algorithms of state transition systems 
such as finite automata and tree automata use a technique of iteratively 
dividing the state space until being congruent with state transitions, 
which can be regarded as simple cases for the construction in this paper. 

Predicate abstraction \cite{Graf1997} is a standard abstraction method 
in software model checking. This method divides an infinite state space 
by introducing appropriate number of predicates that serve as state components 
and determining the state transitions between subspaces using weakest preconditions.
CEGAR (Counterexample-Guided Abstraction Refinement) \cite{Clarke:2003:CAR:876638.876643} iterates the above abstraction
by using a pseudo-counterexample
until the abstracted system satisfies a given verification property or 
a real counterexample to the verification property is found. 

Timed automata (TA) is one of the most popular models of timed systems.
The state space of a TA is infinite because a state contains clocks, which are real numbers. 
For model-checking a TA, the (infinite) state space of the TA is transformed into a finite state space
by region construction or zone construction (see Chapter 17 of~\cite{Clarke:2000:MC:332656}).  Those constructions divide the whole state space
into finite number of subspaces so that subspaces are congruent with state transitions. 
These constructions are similar to the construction in this paper although the former only 
concern TAs.

All of the above-mentioned methods do not deal with fixed-point operations. 
LFP (logics with fixed-point operations) refers to a family of logics
which are extensions of first-order logic by
adding least and greatest fixed-point operations. 
Finite model theory for LFP have been investigated in depth (see Chapters 2 and 3 of~\cite{Gradel:2005:FMT:1206819} for example), 
that assumes only {\em finite} models. 
In contrast, PBES was proposed for investigating the model checking problem of 
first-order $\mu$-calculus that assumes {\em infinite} models in general. 

\section{Conclusions}
We have introduced reduced proof graphs and 
have shown that the existence of a proof graph for data-quantifier free 
and disjunctive PBESs coincides with the existence of a reduced proof graph.
The notion of reduced proof graphs is valuable because there exists
a PBES having a finite reduced proof graph but corresponding proof graphs are all infinite.
We also have shown a way to find a reduced proof graph by constructing the dependency space.
From these results, we obtained a method to solve 
data-quantifier free and disjunctive PBESs
characterized by infinite proof graphs.

Removing \emph{data-quantifier free} restriction is one of future works.

\subsection*{Acknowledgements}

We thank the anonymous reviewers very much for their useful comments to improve this paper.  

\nocite{*}
\bibliographystyle{eptcs}
\bibliography{myrefs}

\end{document}